%% file: ms.tex
\documentclass[12pt,a4paper]{article}
\usepackage[utf8]{inputenc}
\usepackage[english]{babel}
\usepackage[T1]{fontenc}
\usepackage{natbib}
\usepackage{amsmath,amsthm}
\usepackage{amsfonts}
\usepackage{amssymb}
\usepackage{amsthm}
\usepackage{framed}
\usepackage{pdfpages}
\usepackage{subcaption}
\usepackage{tikz}

\usepackage{hyperref}
\newtheorem{theorem}{Theorem}

\newtheorem{corollary}[theorem]{Corollary}

\newtheorem{lemma}[theorem]{Lemma}

\newtheorem{proposition}[theorem]{Proposition}
\newtheorem{assumption}{Assumption}

\begin{document}

\author{Jean-Gabriel Lauzier\footnote{We would like to thank Nenad Kos, Massimo Marinacci and Ruodu Wang for their support and Mario Ghossoub, David Salib, Richard Peter and the participants of the 2021 \textit{American Risk and Insurance Association}'s conference for their comments. We acknowledge financial support from Bocconi University.} \\ University of Waterloo}

\title{Insurance design and arson-type risks}

\maketitle

\begin{abstract}
We design the insurance contract when the insurer faces arson-type risks. The optimal contract must be manipulation-proof. It is therefore continuous, it has a bounded slope, and it satisfies the no-sabotage condition when arson-type actions are free. Any contract that mixes a deductible, coinsurance and an upper limit is manipulation-proof. We also show that the ability to perform arson-type actions reduces the insured's welfare as less coverage is offered in equilibrium.
\end{abstract}


\break
\section{Introduction}
    \input{introduction02}

\section{Model}
    \input{model02}

\section{Discussion and conclusion}
    \input{conclusion02}
    
\appendix
\section{Omitted proofs}
\input{appendixA}
\nocite{*}
\bibliographystyle{ormsv080_modified} 
\bibliography{bibliography.bib}

\end{document}

%% file: introduction02.tex
Suppose Bob was just involved in a bicycle accident. After the fact, an officer of the law provided Bob with a certificate indicating that the automobile driver was responsible for the accident. However, the certificate does not specify how bad the damage inflicted to Bob's steel steed was. Which types of insurance contracts will incentivize Bob to take a sledgehammer to his bicycle before taking a picture and filling his insurance claim?\\

An arson-type action is the action of inflating an insurance claim by physically destroying an object without being caught. The ability of an insured to perform arson-type actions largely impacts the types of contract supplied in equilibrium. We show in this paper that the optimal contract never incentivizes the insured to perform arson-type actions, so that it is manipulation-proof. As such, the ability to perform arson-type actions hurts the insured because it reduces the amount of coverage offered in equilibrium. Manipulation-proofness requires the optimal contract to be continuous and to have a bounded slope, and our model implies the \textit{no-sabotage condition} of \citet{carlier2003pareto} when increasing the damage is free (Bob already has a mace). The types of contract most frequently observed in the real world are robust to arson-type risks. These include all contracts mixing deductible, coinsurance and an upper limit.\\

In other words, Bob will never want to pay for a protection that lets him destroy his bike and profit. Bob's protection is thus somewhat limited, to the extent that he would benefit from committing not to buy a sledgehammer. Of course, Bob is not credible, and he is only offered contracts with limited protection. Since all contracts mixing deductibles, coinsurance, and upper limits are manipulation-proof, we can assume as a rule of thumb that that's what Bob will purchase.\\

Our characterization of the optimal contract exploits a novel argument which is remarkably general. It relies on a rigorous characterization of the manipulation stage of the game induced by a given contract. Formally, we show that for every contract that induces arson-type actions, there exists an alternative contract that provides state-by-state the same protection and does not induce manipulations. Since all arson-type risks are priced in the cost of the contract with a higher premium, the latter contract is cheaper than the former and dominates it. The insurer only offers manipulation-proof contracts in equilibrium. We thus interpret our model as a behaviourally sound justification for assuming the no-sabotage condition. As such, our model suggests that all markets subject to arson-type risks are comonotonic markets \citep{boonen2020competitive}. We can interpret our manipulation-proofness result as a theoretical justification for the heuristic of using "simple" contracts, those contracts that only involve a combination of deductibles, coinsurance, and upper limits. We thus suggest that our model rationalize the common practice of restricting the insurance offer to simple contracts, at least in the field of Property \& Casualty insurance.\\

The rest is as follow. We first close this section with a review of the literature. The core of the article is split in two. We introduce in section 2.1 a general model of insurance design with arson-type risks and show that the optimal contract must be manipulation-proof. This yields a general representation of the optimal contract, which we then use in section 2.2 to find closed form solutions. We conclude by discussing the model's generality and its interest for the risk professionals.

\subsection*{Related literature}

We are not the first to point out that the possibility of inflating an insurance claim by physically destroying an object imposes structure on the type of contract which an insurer can offer. \citet{huberman1983optimal} contains the earliest mention of arson-type actions we are aware of. The authors analyze the optimal insurance contract when respecting the contract involves administrative costs but where there are economies of scale. Their model's optimal first-best contract is a completely disappearing deductible [Fig. 1(a)]. They show that if the insured can cause extra damage, then the second-best contract is a straight deductible [Fig. 1(b)]. Similarly, \citet{picard2000design} introduces arson-type actions to restrict the set of feasible contracts in a setting where the insured can defraud the contract and manipulate the audit costs. The model's first-best contract is discontinuous, and the author shows that this discontinuity disappears when there are arson-type risks.\\
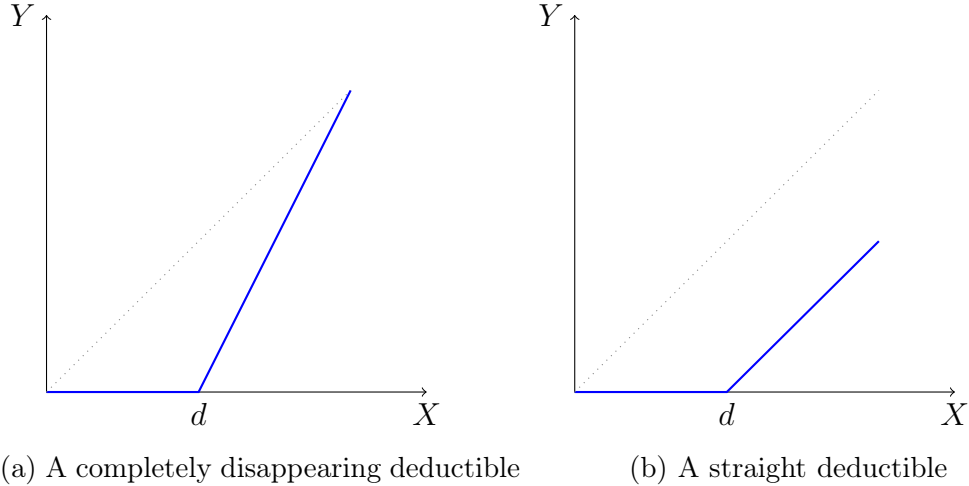
\begin{figure}[ht]
\caption{Disappearing v.s. straight deductible}
\begin{subfigure}{0.5\textwidth}
     \begin{tikzpicture}
        \draw[<->] (0,5) -- (0,0) -- (5,0);
        \draw[blue,thick](0,0) -- (2,0) -- (4,4); 
        \draw[gray,dotted] (0,0) -- (4,4);
        \node[below] at (5,0) {$X$};
        \node[left] at (0,5) {$Y$};
        \node[below] at (2,0) {$d$};
    \end{tikzpicture}  
    \caption{A completely disappearing deductible}
    \label{figure1:sub1}
\end{subfigure}
\begin{subfigure}{0.5\textwidth}
    \begin{tikzpicture}
        \draw[<->] (0,5) -- (0,0) -- (5,0);
        \draw[blue,thick](0,0) -- (2,0) -- (4,2); 
        \draw[gray,dotted] (0,0) -- (4,4);
        \node[below] at (5,0) {$X$};
        \node[left] at (0,5) {$Y$};
        \node[below] at (2,0) {$d$};
    \end{tikzpicture}  
    \caption{A straight deductible}
    \label{figure1:sub2}
\end{subfigure}
\end{figure}

Following the aforementioned articles it has become routine in the literature to assume that the retention schedule is monotonic. For instance,  \citet{cai2020optimal}'s survey contains an exhaustive overview of the use of this (co)monotonicity assumption in the reinsurance literature. This assumption is now referred to as the \textit\textit{no-sabotage condition} as of \citet{carlier2003pareto}. The name refers to the observation that non-monotonic retention schedules seem to incites the insured to inflate its losses. Our model provides a formal counterpart to such argument, thereby shedding lights on its veracity and limitations. We obtain the contract's continuity, its bounded slope and the no-sabotage condition as implications of the same result. Formally, since the contract can be discontinuous, as in \citep{picard2000design}, we cannot use standard first-order conditions to characterize the optimal manipulation correspondence. But the manipulation stage of the game is a positioning choice problem, a class of optimisation problems we defined and characterized in \citep{Lauzier2019positioningmaths}. The ad-hoc envelope theorem of \citet{Lauzier2019positioningmaths} thereby guarantees that the value function of the manipulation stage's problem is continuous and has a bounded slope. This articles main result follows from observing that any contract is dominated by the value function of the optimisation problem it defines in the manipulation stage of the game. The no-sabotage condition obtains as a special case, being the situation where there are no extra costs to inflating the losses. This articles thus complements the literature on the desirability of comonotonic contracts [\citep{landsberger1994co}, \citep{dana2003modelling}, \citep{LUDKOVSKI20081181}, \citep{CARLIER2012207}] by showing how comonotonicity also naturally obtains as the equilibrium outcome of a contracting game.\\

Many types of contracts satisfy the no-sabotage condition and are therefore manipulation-proof with regards to arson-type actions.\footnote{This includes full insurance contracts, straight deductibles and pure coinsurance contracts. See section 2.1.1.} Thus, if a contract is a first-best solution to a problem of insurance design without arson-type actions, it is also a second-best solution to the same problem with arson-types actions. In other words, the possibility of arson-type actions affects the design of insurance contracts if, and only if, the first-best contract is not itself manipulation-proof. Such departures routinely happen when the insurer faces administrative cost to respect the contract. We therefore cast our model in the context of \citet{spaeter1997design}, which contains a general analysis of the design of insurance contracts with administrative costs. This approach allows us to streamline proofs greatly and thus the exposition.\\

 While \citet{Lauzier2019securitydesign} obtained acceptable manipulations while designing securities, we cannot fathom them in the case of insurance contracts. In insurance, the possibility of arson-type manipulations simply hurts the insured by lowering the coverage offered by the insurer in equilibrium. Since all contracts mixing deductibles, coinsurance, and upper limits are robust to arson-type risks, there does not seem to be a trade-off between the prevention of arson-type manipulations and the provision of incentives to prevent ex-ante moral hazard.  This is exactly the opposite result as that obtained in \citet{Lauzier2019securitydesign}. We discuss further this observation in the conclusion.

%% file: model02.tex
We prove in Section 2.1 that the optimal contract must be manipulation-proof, meaning that the optimal contract never induces arson-type actions in equilibrium. Manipulation-proofness implies that the optimal contract is continuous and has a bounded slope, i.e. the optimal contract is a generalized deductible. We turn to closed-form solutions of the model in Section 2.2. We first show that arson-type risks are irrelevant when the administrative cost to deliver the contract is nil. We then study the case of fixed administrative costs to process claims. We show that the ability to use arson-type actions strictly hurts the insured. This is because, in equilibrium, the optimal insurance contract provides less coverage than it would otherwise. We conclude with an analysis of the second-best contract when the first-best is a completely disappearing deductible.

\subsection{Notation and manipulation-proofness}
Let $S$ be a set of states of the world, let $\mathbb{P}$ be a probability measure for states $s\in S$. Let the risk $X:S \rightarrow [0,M]$ be a continuous random variable with full support $[0,M]$ and some mass at
$0$, this mass being the probability that no accident occurs. 
Let the function $c:[0,M] \rightarrow \mathbb{R}_+$ be the administrative cost of
respecting the insurance contract and let $Y:[0,M]\rightarrow [0,M]$ be the
indemnity schedule.\footnote{So $Y\in
B_+(\mathcal{B}([0,M]))$, the space of non-negative and bounded functions (sup-norm) which
are measurable with regard to the Borel $\sigma$-algebra $\mathcal{B}$ of $[0,M]$.} The
amount $H\geq 0$ denotes the price (premium) of the contract while $W_0>0$ is the initial
wealth of the insured and $\rho \geq 0 $ is the loading factor.
As usual, the function $u:\mathbb{R}_+ \rightarrow \mathbb{R}_+ $ is a twice
differentiable and strictly concave Bernoulli utility function satisfying Inada
conditions.\\

\noindent The game proceeds as follows:

\begin{description}
    \item[Stage 1] the insured buys the insurance contract $Y$ at price $H$;
    \item[Stage 2] the state $s$ realizes and loss $X(s)$ occurs (Nature moves);
    \item[Stage 3] the insured observes the loss and decides to take hidden action
$z\in [0,M - X(s)]$ to augment the damages;
    \item[Stage 4] the contract is implemented without renegotiation.
\end{description}

\noindent The solution concept is a weak Perfect Bayesian equilibrium \citep{mas1995microeconomic} where we assume
that the insured takes the insurer's favoured action whenever indifferent.\\

By backward induction the optimal contract solves the following
optimisation program: 
\begin{align}
    \sup_{H \geq 0, Y \in B_+(\mathcal{B}([0,M]))} & \int u(W_0 - H - X(s)- z(s) -g(z(s))+Y(X(s)+z(s)))d\mathbb{P} \tag{Problem I} \label{Problem I}\\
      s.t.\, & \,0 \leq Y \tag{LL} \label{LL}\\
      &\, \forall s,\, Y(X(s)+z(s)) \leq X(s)+z(s) \tag{B} \label{B}\\
           & (1+\rho)\int Y(X(s)+z(s))+c(Y(X(s)+z(s))) d \mathbb{P} \leq H \tag{PC} \label{pci} \\
           & \forall s,\,  z(s) \in \arg\max_z\{ Y(X(s)+z)- z -g(z)\} \tag{IC} \label{ic}
\end{align}
where \eqref{LL} is the insured's limited liability constraint,
\eqref{B} is the "boundedness constraint" stating that the insurer will never
pay more than the observed loss, \eqref{pci} is the insurer's participation constraint,
\eqref{ic} is the insured's incentive compatibility constraint
and the function
\begin{align*}
    g(z)= \begin{cases} +\infty &\text{ if } z<0\\ \beta z &\text{ if } z\geq 0 \end{cases}
\end{align*}
represents an extra cost of inflicting damage (Bob buying a sledgehammer).\\

We aim to prove that the optimal contract is Lipschitz continuous with constant $\leq 1+\beta$. Let us start with a handy assumption which is common in the literature:
\begin{assumption} We assume throughout that feasible contracts are non-decreasing and upper semi-continuous functions. \end{assumption}
Of course, this implies that the optimal contract $Y$ is non-decreasing and upper semi-continuous (if it exists). Assumption 1 is standard and need not be discussed further. Similarly, the next statement is standard and will not be proved:
\begin{lemma} The insurer's participation constraint \eqref{pci} must be binding.\end{lemma}

Lemma 1 simply states that the contract must be sold at actuarially fair price. Let us now inspect the incentive compatibility constraint \eqref{ic}. This constraint is an optimisation program which is ill-behaved. This is because we do not know at this level of generality if the optimal contract $Y$ is continuous. Moreover, we cannot assume away the possibility that $Y$ is discontinuous. We seem to be in trouble now, because we cannot differentiate the objective function in \eqref{ic} and thus cannot use standard first-order conditions to characterize the (set of) optimal manipulations. Fortunately, the optimisation problem of constraint \eqref{ic}
is what we defined as a positioning choice problem in \citet{Lauzier2019positioningmaths}. The interest of positioning choice problems relies in that their value function is always Lipschitz continuous and almost everywhere differentiable. The next ancillary lemma is an immediate consequence of the ad-hoc envelope theorem of \citet{Lauzier2019positioningmaths}.\\

Let the value function $V$ of the manipulation stage of the game be \begin{align*}
    V(s;Y)= Y(X(s)+z(s))- z(s) -g(z(s))
\end{align*}
for 
\begin{align*}
    z(s) \in \sigma(s;Y):=\arg\max_z\{ Y(X(s)+z)- z -g(z)\},
\end{align*}
where $\sigma$ denotes the optimal choice correspondence of the manipulation stage of the game. 
\begin{lemma}
 If $Y$ is non-decreasing and upper semi-continuous then $V$ is Lipschitz continuous with constant $\leq 1+\beta$ and almost everywhere differentiable.
\end{lemma}

We are now ready to prove our main result. We say that a contract $Y$ is manipulation-proof if for every $s \in S$ it is $$0 \in \sigma(s;Y).$$
\begin{theorem} Any optimal contract $Y$ is manipulation-proof.\end{theorem}
\begin{proof}
    Suppose, by contraposition, that $Y$ is optimal but that there exist a set $S'\subset S$ of such that simultaneously $\mathbb{P}[s\in S']>0$ and for every $s\in S'$ it is 
$$0\notin \sigma(s;Y).$$
Let $V(s;Y)$ denote the value function of the manipulation problem defined by $Y$, and consider now the alternative contract $(\overline{H}, \overline{Y})$ where $\overline{Y}=V$, i.e. the new indemnity schedule gives state-by-state the same final reimbursement than the original one after manipulations. The new indemnity schedule $\overline{Y}$ is manipulation-proof: by Assumption 1 and Lemma 2 it holds that $\overline{Y}$ is non-decreasing and Lipschitz with constant $1+\beta$, and thus to any two $x,x'\in[0,M]$ such that $x<x'$ it is 
$$\overline{Y}(x) \geq \overline{Y}(x') - g(x'-x),$$
and for every realisation $x\in [0,M]$ we have
$$0 \in \sigma(x;\overline{Y}):=\arg\max_z\{\overline{Y}(x+z) -g(z) \}.$$
The new contract strictly dominates the original contract. Indeed, notice that 
\begin{align*}
    \int Y(X(s)+z(s))+c(Y(X(s)+z(s))) d \mathbb{P} > \int \overline{Y}(X(s)) + c(\overline{Y}(X(s))d\mathbb{P}
\end{align*}
so the price $\overline{H}$ of $\overline{Y}$ is strictly smaller than the price $H$ of the original contract.
\end{proof}

We can understand Theorem 3 as stating that since the insurer fully prices arson-type risks, contracts which induce arson-type actions will never be offered in equilibrium. Indeed, all extra damage due to arson-type actions must lead to a higher premium for the contract to be actuarially fair. Of course, the insured will prefer the cheapest contract as he receives state-by-state the same final amount under both contracts. Or, from Bob's perspective, he was offered two contracts offering the same protection. An expensive one which would allow him to take a sledgehammer to his bike and a cheaper one which did not, so Bob, being a rational person, chose the cheaper option.
\begin{corollary} Any optimal contract $Y$ must be Lipschitz and with slope $\leq 1+\beta$.\end{corollary}

\subsubsection{Characterization}
Corollary 4 implies that the family of contracts
$$\{Y\in B_+(\mathcal{B}([0,M])):Y(s)=V(s)\}$$
consists of functions which are a.e. differentiable.\footnote{This follows from Rademacher's Theorem. See Theorem 3 of \citet{Lauzier2019positioningmaths}.} We can thus rewrite \eqref{Problem I} as
\begin{align}
    \max_{H \geq 0, Y \in B_+(\mathcal{B}([0,M]))} & \int u(W_0 - H - X(s) +Y(X(s))d\mathbb{P} \tag{Problem S} \label{Problem S}\\
      s.t.\, & \,0 \leq Y \leq X \tag{S1} \label{S1}\\
            & slope(Y)\leq 1+ \beta \tag{S2} \label{S2}\\
           & (1+\rho)\int Y(X(s))+c(Y(X(s))) d \mathbb{P} = H \tag{S3}\label{S3}
\end{align}
under the implicit assumption that $Y\in C^0[0,M]$ is a.e. differentiable.\footnote{The notation $C^0$ denotes the space of
continuous functions.} Notice how this problem, as rewritten, is almost identical to the problem studied in \citet{spaeter1997design}, except for the extra constraint $slope(Y)\leq 1+\beta$.\\

We conclude this section with a general representation theorem and a few observations that are handy when deriving closed form solutions. As a consequence of Theorem 3, all optimal contract $Y$ can be written as
\begin{align*}
    Y(x)=\max\{0, \alpha(x)x-d\},
\end{align*}
where $d\geq 0$ is a deductible and $\alpha(x)$ is a non-negative, continuous and a.e. differentiable function satisfying for every $x\in [0,M]$
\begin{align*}
    0 \leq \frac{\partial \alpha(x)x}{\partial x}\leq 1+\beta.
\end{align*}
We say that contract $Y$ 
\begin{itemize}
    \item is a \textbf{full insurance contract }when $d=0$ and $\alpha(x)=1$ everywhere;
    \item is a \textbf{straight deductible} when $d>0$ and $\alpha(x)=1$ everywhere;
    \item entails \textbf{coinsurance} when $\alpha(x)=\alpha \in (0,1)$ and
    \item has \textbf{upper limit} $\delta>0$ if for every $x$, $Y(x)\leq \delta$, with strict equality for some $x$.
\end{itemize}
Any contract mixing deductibles, coinsurance and upper limits can be written as
\begin{align*}
    Y(x)=\min\{\delta, \max\{0, \alpha x - d\}\}.
\end{align*}

When $\beta=0$ we recover the well-known \textit{no-sabotage condition} explained in details in \citet{carlier2003pareto}. Formally, if $\beta=0$ then for any contract $Y$ the following holds:
\begin{enumerate}
    \item $slope(Y) \leq 1$;
    \item the retention function $R(x)=x-Y(x)$ is weakly monotone increasing;
    \item the random vector $(X,Y, X-Y)$ is comonotonic.
\end{enumerate}

\subsection{Some closed form solutions}
We now provide closed form solutions to our contracting problem.
\subsubsection{Full insurance, straight deductibles and the irrelevance of arson}
\begin{lemma}[Arrow's Theorem]
If $c=0$ the optimal contract is a straight deductible $Y(x)=\max\{0,x-d\}$ for which $d=0$ if, and only if, $\rho=0$.
\end{lemma}
Lemma 5 need not be proved as it is the classic Arrow-Borch-Raviv Theorem (see \citet{dionne2000handbook}). The lemma tells us that \eqref{Problem S} becomes interesting only when $c>0$ somewhere. This is because when $c=0$ the first-best contract never incentivizes Bob to take a sledgehammer to his bike. In other words, arson-type risks impact the provision of insurance only when there is a meaningful reason not to provide either full insurance or a straight deductible.

\subsubsection{Fixed costs and nuisance claims: arson-type risks reduce welfare}
Let us now consider the case when the administrative costs to deliver the contract are fixed. This situation is important because it is the easiest way to show that arson-type risks hurt the insured, i.e. that Bob would prefer to be unable to destroy his bike. 
\begin{assumption}
The cost function involves only a fixed cost per claim: $c$ satisfies $c(0)=0$ and $c(y)=c_0>0$ for every claim requiring $y>0$ to be paid by the insurer.
\end{assumption}
We claim that when the insured can freely augment the damage ($\beta =0$) then the optimal contract is a straight deductible.
\begin{proposition}
Under assumption 2, if $\beta =0$ then the solution to \eqref{Problem S} is a straight deductible: there exist a $d>0$ such that
$$Y(x)= \max\{0, x-d\}.$$
\end{proposition}
\begin{proof}
Notice that the insured files a reclamation if, and only if, $Y(x)>0$, so we can partition the interval $[0,M]$ in two regions $\mathcal{M},\mathcal{M}^C$ such that
\begin{align*}
    &\mathcal{M}=\{x\in [0,M]: Y(x)>0\}=\{x: \text{ the insured files a claim}\}\\
    &\mathcal{M}^C=\{x\in [0,M]: Y(x)=0\}=\{x: \text{ the insured does not files a claim}\}.
\end{align*}
Recall that our optimal contract takes the form of a generalized deductible
\begin{align*}
    Y(x)=\max\{0, \alpha(x)x-d\}.
\end{align*}
Since the marginal cost $c'$ is null it is easily verified that \begin{align}
    \left. \frac{\partial Y(x)}{\partial x}\right\vert_{x\in \mathcal{M}} \geq 1.
\end{align}
However, since $\beta =0$ constraint \eqref{S2} implies that $Y'\leq 1$ and thus the previous is an equality, i.e. the optimal contract provides full marginal insurance on $\mathcal{M}$. We can thus set $\alpha (x)=1$ and rewrite problem \eqref{Problem S} as
\begin{align}
\max_{d\geq 0} & \int u(W_0 - H - X(s) +\max\{0, X(s)-d\})d\mathbb{P} \tag{Problem FC \& A} \label{Problem FC A} \\  
s.t.\quad  & c_0 \mathbb{P}[X\geq d] + \mathbb{E}[X-d\vert X\geq d] = (1+\rho)^{-1}H  \nonumber.
\end{align} \end{proof}


We want to show that the possibility of arson-type risks hurts the insured, i.e that Bob would like to commit to leave his bike intact because he would be offered a better contract. Formally, we will show that absent the incentive compatibility constraint, a contract with an upward jump would improves upon a straight deductible. Let us consider a contract of the form
\begin{align*}
    Y_{Ret}(x;t,j)=\begin{cases} 0 &\text{ if } x<t \\ x-j &\text{ if } x\geq t \end{cases}
\end{align*}
for parameter $t\geq 0$ being a threshold, $j\geq 0$ being a loss retention parameter and $t\geq j$. The number $t-j$ is the magnitude of the jump of $Y_{Ret}$ at $t$ and is interpreted as the minimum amount of money that the insured receives provided he files a claim. We will refer to contracts with the above form as \textit{contracts with constant retention}. Clearly, $Y_{Ret}$ is not manipulation-proof when $t>j$ and $Y_{Ret}$ is a simple deductible when $t=j$. More importantly, notice that the insured files a reclamation if, and only if, $Y_{Ret}(x)>0$, so we can again partition the interval $[0,M]$ in the two regions $\mathcal{M}$ and $\mathcal{M}^C$ defined above.\\

\begin{figure}[ht!]
\caption{Straight deductibles and constant retention contracts}
\begin{subfigure}{0.5\textwidth}
    \begin{tikzpicture}
        \draw[<->] (0,5) -- (0,0) -- (5,0);
        \draw[blue,thick](0,0) -- (2,0) -- (4,2); 
        \draw[gray,dotted] (0,0) -- (4,4);
        \node[below] at (5,0) {$X$};
        \node[left] at (0,5) {$Y$};
        \node[below] at (2,0) {$d$};
    \end{tikzpicture}  
    \caption{Straight deductible}
    \label{figure2:sub1}
\end{subfigure}
\begin{subfigure}{0.5\textwidth}
    \begin{tikzpicture}
        \draw[<->] (0,5) -- (0,0) -- (5,0);
        \draw[blue,thick](0,0) -- (2,0);
        \draw[blue,thick](2,1) -- (4,3);
        \draw[gray,dotted] (0,0) -- (4,4);
        \draw[green,red] (2,0)--(2,1);
        \node[red, right] at (2,0.5){$t-j$};
        \node[below] at (5,0) {$X$};
        \node[left] at (0,5) {$Y$};
        \node[below] at (2,0) {$t$};
    \end{tikzpicture}  
    \caption{Constant retention contract}
    \label{figure2:sub2}
\end{subfigure}
\end{figure}
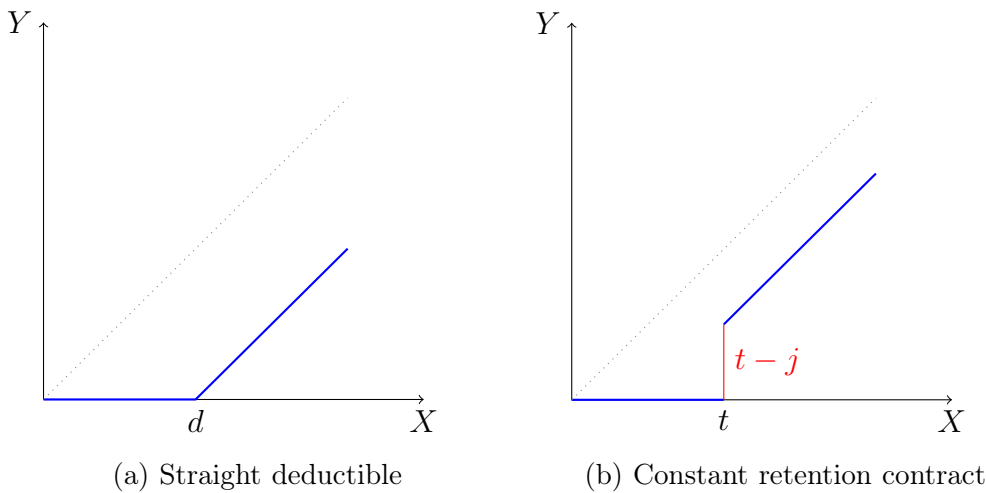
Suppose now that the cost $c$ satisfies Assumption 2, that the insured cannot use arson-type actions so that there are no incentive compatibility constraints and that we are restricting our search to contracts with constant retention. A few calculations shows that we can write the optimisation problem as 
\begin{align}
\max_{t\geq 0, j\geq 0} & \int u(W_0 - H - X(s) +Y_{Ret}(X(s);t,j)d\mathbb{P} \tag{Problem FC \& No-A} \label{Problem FC No-A} \\  
s.t.\quad  & c_0 \mathbb{P}[X\geq t] + \mathbb{E}[X-j\vert X\geq t] = (1+\rho)^{-1}H  \nonumber\\
& t-j\geq 0 \tag{Positive jump} \label{positive jump}.
\end{align}
Clearly, the optimisation problem with arson-type risks \eqref{Problem FC A} is \eqref{Problem FC No-A} when $t=j$. It is thus clear that arson-type risks can hurt the insured. We are left to show that arson-type risks always hurt the insured when the fixed cost is strictly positive. This entails to show that the constraint \eqref{positive jump} of \eqref{Problem FC No-A} is an equality if, and only, $c_0>0$. But this is a well-known result, the statement being essentially Theorem 2 of \citet{gollier1987}.

\begin{lemma}
Consider \eqref{Problem FC No-A}. It holds that $t=j$ if, and only if, $c_0=0$.
\end{lemma}

The proof does not provide intuition and is thus omitted from the text. The interpretation is straightforward. Absent arson-type risks, the optimal insurance contract provides the maximum possible protection while eliminating nuisance claims, the small claims which are more costly to administer than the coverage they offer. This is achieved by refusing to reimburse small claims while offering a generous protection conditional on the loss being large enough. This creates a spread in coverage, the discontinuity at the threshold. The problem is that Bob would like to exploit this spread by augmenting the damage to his bike. In equilibrium, Bob is never offered such contract with high protection, and Bob would be strictly better off if he could commit not to buy a mace.

\subsubsection{Continuous costs and the sub-optimality of disappearing deductibles}
We now consider the case when $c$ is a continuous function. 
\begin{assumption}
The cost function $c$ is continuous, weakly positive, non-decreasing, twice-differentiable and satisfies $c(0)=0$.
\end{assumption}
Recall that our problem is now almost identical to the one of \citet{spaeter1997design}.
\begin{lemma} Under Assumption 3, if $Y$ solves the reduced problem
\begin{align}
\max_{H \geq 0, Y \in B_+(\mathcal{B}([0,M]))} & \int u(W_0 - H - X(s) +Y(X(s))d\mathbb{P} \tag{Reduced problem} \label{reduced problem}\\
  s.t.\, & \,0 \leq Y \leq X \nonumber \\
       & (1+\rho)\int Y(X(s))+c(Y(X(s))) d \mathbb{P} = H  \nonumber
\end{align}
of \citet{spaeter1997design} and $Y$ satisfies constraint \eqref{S2}
then $Y$ solves \eqref{Problem S}.
\end{lemma}

This lemma informs us that the only problematic case which must be handled is when the
contract $Y$ found in \citet{spaeter1997design} does not satisfy constraint
\eqref{S2} somewhere. Intuitively, it seems natural to attempt flattening $Y$ sufficiently to 
satisfy $slope(Y)\leq 1+\beta$ thus solving \ref{Problem S}. This approach is sometimes fruitful but does not work in certain cases as we explain later. We consider the easiest case when the best contract absent arson-type risks is a \textit{completely} disappearing deductible.\footnote{\citet{huberman1983optimal} is the first to notice that the optimality of disappearing deductibles is no longer true when the insured has access to arson-type actions.} Formally, we say that contract $Y$ is a completely disappearing deductible if there exists a realisation $x'\in[0,M]$ such that for every $x\geq x'$ it is $Y(x)=x$.\\

\begin{proposition}
If $\beta=0$, $Y_R$ solves \eqref{reduced problem} and $Y_R$ is a
  completely disappearing deductible then the unique solution to \eqref{Problem I} is a straight deductible.
\end{proposition}
\begin{proof}
The fact that $Y_R$ is a completely disappearing deductible informs us that our optimal contract $Y$ must provide as much coverage as possible. Since $\beta=0$ it is $$\left. slope(Y)\right\vert_\mathcal{M} =1,$$
for $\mathcal{M}$ being as before.  Hence, finding $Y$ is equivalent to solving
\begin{align*}
\max_{d\geq 0} & \int u(W_0 - H - X(s) +\max\{0, X(s)-d\})d\mathbb{P} \\  
s.t.\quad &(1+\rho)\mathbb{E}[\max\{0, X-d\} +c(\max\{0, X-d\})] = H  \nonumber.
\end{align*} 
\end{proof}

Similarly as before, Bob's ability to take a sledgehammer to its bicycles is 
priced in the insurance contract, and Bob will never be offered
a completely disappearing deductible in equilibrium. Again, though Bob would like to purchase such contract, he cannot. This is because no Bob
can commit to being honest. The intuition that the contract solving problem \eqref{Problem S} is
simply a "flattening" of the solution to \eqref{reduced problem} is misleading.
By additivity of the Lebesgue integral, the insurer's participation constraint
states that the insurer should recoup the cost on average and not state-by-state. This, fortunately, gives us some leeway in solving
\eqref{Problem S}. However, this also means that there are some cases
where we have to roll-up our sleeves and directly attack the problem.

%% file: conclusion02.tex
We showed that the optimal contract must be robust to arson-type risks and provided a general characterization of the contract. We now explain why this result is remarkably general and why we believe both risk theorists and practitioners alike should be interested in it. We focus on risk practitioners, whom might want to consider more complex situations than what is envisioned in the text. The main message of the paper is not to worry. The heuristic of considering only contracts mixing deductibles, coinsurance and upper-limits is well-founded when there are arson-type actions, and the risk professionals might want to focus their attention on controlling other risks. A simple example is a car insurer that uses coinsurance clause to prevent reckless driving (ex-ante moral hazard, see below). As coinsurance contracts are manipulation-proof with regard to arson-type actions, the car insurer should not worry about arson-type risks and can thus focus on pricing effectively reckless driving. We conclude with a discussion of the limitations of the model that are relevant to the practitioners and with avenues for further research.

\subsection*{Non-expected utility decision criteria}
A careful inspection of the proof of manipulation-proofness (Theorem 3) reveals two important facts. First, notice that we never used any properties of the underlying probability space beside that (i) the random variable representing the risk faced by the insured had full support and that (ii) the two integrals in \eqref{Problem I} are finite. In layman's terms, we only used the fact that when Bob has an accident (i) the damage to his bike can fully range from small scratches to a total loss and (ii) that Bob's risk is not infinitely large. This is important because we could have considered other decision criteria without changing anything to the veracity of Theorem 3. For instance, replacing the probability $\mathbb{P}$ with a capacity $\nu$ in one integral and integrating in the sense of Choquet instead of Lebesgue would have changed nothing, provided that the set of solutions to \eqref{Problem I} remains non-empty and another technical condition is satisfied. This is particularly important for the risk practitioners because it implies that one does not want to sell insurance contracts vulnerable to arson-type actions no matter how difficult it is to evaluate the fundamentals of the risk to be bear.\footnote{The technical condition is that the risk to insure must still have full support under the probabilistic belief of the decision-maker with the non-expected utility criterion. The optimal contract might not be manipulation-proof otherwise. This happens when a decision-maker believes that the set of realizations for which a contract induces manipulation is a set of measure zero.}\footnote{Non-expected utility decision criterion can be interpreted as situations where the decision-maker has to statistically infer the risks. This extension does not seem to change the fundamental message of the article. In particular, it seems that one can keep invoking arson-type actions to justify the no-sabotage condition.}

\subsection*{Ex-ante moral hazard in loss reduction}

The most well-known and studied agency problem is the classical Principal-Agent problem with hidden actions, or ex-ante moral hazard.\footnote{We refer to \citet{winter2013optimal} for an introduction to the Principal-Agent problem in the context of insurance design.} In the context of insurance, the hidden action is often interpreted as unobservable preventive measures that the insured takes to reduces its risk, for instance by driving carefully. There does not seem to be a trade-off in the provision of incentives to mitigate ex-ante moral hazard and to prevent arson-type actions. This is because the optimal contract with only ex-ante moral hazard (often) satisfies the no-sabotage condition and his therefore manipulation-proof with regards to arson-type actions. Precisely, the contract satisfies the no-sabotage condition when the \textit{first-order approach} is valid, a condition most often assumed in the literature.

\subsection*{Limitations and avenues for further research}

The result that optimal contract is manipulation-proof seems to be in contradiction with the real-world, where arson and insurance fraud does happen. It is not. This apparent contradiction comes from the interplay of three important assumptions:  we defined arson-type actions as (i) the ability to physically destroy an object while (ii) having no chance of being caught and while (iii) assuming implicitly that the insured faces no other risks.\\

 While (i) is a natural consideration in the context of Property insurance, it does not capture all the possibilities of insurance fraud. On theoretical grounds, the classical costly state falsification model of \citep{crocker1998honesty} contains an intuitive example where manipulations happen. In their model, the insurer can defraud the contract by declaring greater damages than the real damages, and the insurer cannot verify the claims. The optimal contracts always entail manipulations in equilibrium because the marginal cost of lying is essentially null for small lies. However, the insured's ability to fill dishonest claims depends on the damage's size, and the insurer can still differentiate between small and large fundamental losses. This unravelling implies that a contract with manipulations in equilibrium dominates (their model's) manipulation-proof contracts.\footnote{The topic of fraud has received a great deal of attention in the literature, and our short discussion is far from exhaustive. We refer to \citet{picard2013economic} for an overview.}\\
 
 With (ii), we implicitly assumed that arson-type actions cannot be detected by auditing. This is unlikely to be true in general, as arson-type actions often leave evidences. For instance, fire accelerants like gasoline leave traces that can be detected by forensic investigators. Given that perfectly preventing arson-type actions by using manipulation-proof contracts is costly in term of welfare, one might want to control for these actions by auditing suspect claims instead of implementing manipulation-proof contracts. This opens up an interesting theoretical possibility, where arson-type actions might happen in equilibrium as a result of optimal contracting.\\

There are situations where arson-type risks does not seem to be realistic on practical grounds. The leading example is the field of health insurance, where arson-type actions as we modelled them would require the insured to commit self-injury. However, anecdotes exist where people attempted to defraud an insurance contract by hurting themselves, health insurance contracts routinely contains self-inflicted injury exclusion clauses, and insurance companies routinely investigates injury claims. It is thus not far fetched to think that self-injuries can be modeled using the method developed in this article. Which brings us to (iii), as our manipulation-proof result does not cover the situation where the insured also faces uninsurable (background) risks. This observation also opens up another interesting theoretical possibility, where arson-types actions like self-injuries could arise in equilibrium as a response to large uninsurable losses.

%% file: appendixA.tex
\noindent \textbf{Proof of Lemma 7}
\begin{proof}
    If $c_0=0$ then \eqref{Problem FC No-A} is the standard Arrow-Borch-Raviv problem for which the solution is a simple deductible, i.e.
    \begin{align*}
        Y(x)=\max\{0, x-d\}
    \end{align*}
so $t=j=d$. Conversely and set $t=j=d^*\geq 0$ such that $d^*$ solve the reduced problem 
\begin{align*}
 \max_{d\geq 0} & \int u(W_0 - H - X(s) +\max\{0, X(s)-d\})d\mathbb{P} \\  
s.t.\quad  & c_0 \mathbb{P}[X\geq d] + \mathbb{E}[X-d\vert X\geq d] = (1+\rho)H  \nonumber.
\end{align*}
If $c_0=0$ then we are done. Suppose $c_0>0$. Then there exists an $\varepsilon>0$ such that the alternative contract
$$Y(x)=\mathbb{I}_{x\geq d^*}[x-d^*+\varepsilon]$$ strictly improves upon the contract
\begin{align*}
        Y(x)=\max\{0, x-d^*\}.
    \end{align*}
    \end{proof}
    